\documentclass[runningheads,a4paper]{llncs} %

\usepackage{amsmath}
\usepackage{amssymb}
\usepackage{graphicx}

\usepackage{url}
\urldef{\mailsa}\path|{furer}@cse.psu.edu|    
\newcommand{\keywords}[1]{\par\addvspace\baselineskip
\noindent\keywordname\enspace\ignorespaces#1}

\newcommand{\im}[2]{#1\langle #2 \rangle}
\newcommand{\et}[2]{\eta_{#1,#2}}
\newcommand{\rh}[2]{\rho_{#1 \rightarrow #2}}
\newcommand{\mer}[1]{\theta_{#1}}

\newcommand{\tw}{\mbox{tw}}
\newcommand{\cw}{\mbox{cw}}
\newcommand{\fw}{\mbox{fw}}

\begin{document}

\mainmatter  %

\title{A Natural Generalization of Bounded Tree-Width and Bounded Clique-Width}

\titlerunning{Generalization of Bounded Tree-Width and Bounded Clique-Width}

\author{Martin F\"urer%
\thanks{Research supported in part by NSF Grant CCF-0964655 and CCF-1320814.}}
\authorrunning{Martin F\"urer}

\institute{Department of Computer Science and Engineering \\
	Pennsylvania State University \\
	University Park, PA 16802,  USA \\
	furer@cse.psu.edu \\
\url{http://cse.psu.edu/~furer}}

\maketitle

\begin{abstract}
We investigate a new width parameter, the fusion-width of a graph. It is a natural generalization of the tree-width, yet strong enough that not only graphs of bounded tree-width, but also graphs of bounded clique-width, trivially have bounded fusion-width. In particular, there is no exponential growth between tree-width and fusion-width, as is the case between tree-width and clique-width. %
The new parameter gives a good intuition about the relationship between tree-width and clique-width.

\end{abstract}

\section{Introduction}
Tree-width is a very natural concept. In an intuitive direct way, it measures how similar a graph is to a tree. Many graph problems are not only easy for trees, but also for other tree-like graphs. Indeed there is a huge number of efficient algorithms for graphs of bounded tree-width.

While graphs of bounded tree-width are sparse, there are some dense graphs, like the complete graph $K_n$ or the complete bipartite graph $K_{nn}$ for which most computational problems have a trivial solution. Graphs of bounded clique-width are intended to cover classes of graphs for which many problems have efficient solutions, even though they contain many dense graphs.

Not unlike tree-width, the concept of clique-width \cite{CourcelleO2000} is based on a type of graph decompositions \cite{CourcelleER93} chosen to allow fast algorithms for large classes of graphs. As the name clique-width indicates, this measure is designed to ensure that complete graphs have a very small width. But neither does clique-width measure a closeness to a clique in a natural way, nor is there an intuitive width (as in tree-width) involved in the definition of this concept. Even though the definition of clique-width is fairly simple, it is harder to obtain an intuition for the classes of graphs with small clique-width.

Bounded clique-width is an extension of bounded tree-width in the sense that every class of bounded tree-width is also a class of bounded clique-width \cite{CourcelleO2000,CorneilR2005}. But the worst case clique-width of graphs of tree-width $k$ has been upper \cite{CourcelleO2000,CorneilR2005} and lower \cite{CorneilR2005} bounded by an exponential in $k$. The fact that this containment result is a difficult theorem suggests that the extension from tree-width to clique-width might not be very natural.

Equivalent to clique-width up to a factor of 2 is the notion of \emph{$k$-NLC (node label controlled) graphs}. Partial $k$-trees have been shown to be $(2^{k+1} - 1)$-NLC trees \cite{Wanke94}. $k$-NLC trees are a sub-class of the $k$-NLC graphs.

In contrast to clique-width (and the width measure produced by NLC trees), we propose a natural generalization of tree-width, which simultaneously generalizes clique-width. Furthermore, containment in the new class is obtained basically without increasing the parameter in both cases. We call the new measure \emph{fusion-width}. We initially choose the name multi-tree-width to emphasize that it is a very natural extension of 
tree-width, by which it is motivated
(even though it is much more powerful). We follow the strong suggestion of two referees to name it differently. The main difference is that while tree-width deals with single vertices or pairs of vertices at a time, fusion-width deals with multiple vertices (with the same label) or pairs of sets of vertices at a time. 

We show that the clique-width grows at most exponentially in the fusion-width, implying that all classes of graphs with bounded fusion-width also have bounded clique-width. Furthermore, for some classes of graphs, there is really such an exponential growth.

Other width parameters generalizing both tree-width and clique-width without blowing up the parameter are rank-width \cite{OumS2006} and boolean width \cite{Bui-XuanTV2011}. The rank-width has the additional nice property of being computable in FPT \cite{HlinenyO2007}. For an overview of width parameters, see \cite{HlinenyOSG08}. There are infinite classes of graphs where the clique-width is exponentially bigger than the boolean width.

The fusion-width has the additional property of being easy to work with and of being the most natural generalization of tree-width and clique-width. This is an essential strength of the new notion of fusion-width.

\section{Definitions}
For the definitions of \emph{FPT (fixed parameter tractable)}, \emph{XP (fixed parameter polynomial time)}, and \emph{tree decomposition}, see e.g. \cite{DowneyF99}.

\begin{definition}
The \emph{tree-width} $\tw(G)$ \cite{RobertsonS84} of a graph $G$ is the smallest integer $k$,
such that $G$ has a tree decomposition with largest bag size $k+1$.
\end{definition}

It is NP-complete to decide whether the tree-width of a graph is at most $k$ (if $k$ is part of the input) \cite{ArnborgCP87}.
For every fixed $k$, there is a linear time algorithm deciding whether the tree-width is at most $k$, and if that is the case, producing a corresponding tree decomposition \cite{Bodlaender96}. For arbitrary $k$, this task can still be approximated. A tree decomposition of width $O(k \log n)$ %
\cite{BodlaenderGHK95} and even $5k$ \cite{BodlaenderDDFLP2013} can be found in polynomial time.

Closely related to tree-width is the notion of \emph{branch-width} \cite{RobertsonS91}.

\begin{definition}
A \emph{$k$-expression} is an expression formed from the atoms  $i(v)$, the two unary operations $\et{i}{j}$ and $\rh{i}{j}$, and one binary operation $\oplus$ as follows.
\begin{itemize}
\item $i(v)$ creates a vertex $v$ with label $i$, where $i$ is  from the set $\{1, \dots , k\}$.
\item $\et{i}{j}$ creates an edge between every vertex with label $i$ and every vertex with label $j$ for $i \neq j$.
\item $\rh{i}{j}$ changes all labels $i$ to $j$.
\item $\oplus$ does a disjoint union of the generated labeled graphs.
\end{itemize}
Finally, the \emph{generated graph} is obtained by deleting the labels.
\end{definition}

\begin{definition}
The \emph{clique-width} $\cw(G)$ of a graph is the smallest $k$ such that the graph can be defined by a $k$-expression \cite{CourcelleER93,CourcelleO2000}.
\end{definition}

Computing the clique-width is NP-hard \cite{FellowsRRS2006}.
Thus, one usually assumes that a graph is given together with a $k$-expression.
For constant $k$, the clique-width can be approximated by a constant factor in polynomial time \cite{OumS2006}, in fact, this factor can be made smaller than 3 \cite{Oum2008}.

\section{The Fusion-Width}
We define a new width measure $\fw(G)$ (fusion-width of $G$) with the properties
\[\fw(G) \leq \tw(G) + 2 \mbox{ and } \fw(G) \leq \cw(G).\]
We want these containments to be obvious.
Still, we would like all tasks known to be solvable in polynomial time for graphs of bounded clique-width (and therefore of bounded tree-width) to be solvable in polynomial time also for graphs of bounded fusion-width.

These inequalities immediately imply that the class of graphs of bounded clique-width is  contained in the class of graphs of bounded fusion-width.
The definition of fusion-width is obtained as a simple extension of the definition of clique-width by a new operation $\mer{i}$, merging all vertices with label $i$.

\begin{definition}
 A \emph{$k$-fusion-tree expression} is an expression formed from the atoms  $\im{i}{m}$, the three unary operations $\et{i}{j}$, $\rh{i}{j}$, and $\mer{i}$, and one binary operation $\oplus$ as follows.
\begin{itemize}
\item $\im{i}{m}$ creates a graph consisting of $m$ isolated vertices labeled $i$, where $i$ is  from the set $\{1, \dots , k\}$.
\item $\et{i}{j}$ creates an edge between every vertex with label $i$ and every vertex with label $j$ for $i \neq j$. 
\item $\rh{i}{j}$ changes all labels $i$ to $j$.
\item $\mer{i}$ merges all vertices labeled $i$ into one vertex. The new vertex is labeled $i$ and is adjacent to every vertex not labeled $i$ to which some vertex labeled $i$ was adjacent before the operation.
\item $\oplus$ does a disjoint union.
\end{itemize}
Finally, the \emph{generated graph} is obtained by deleting all the labels.
\end{definition}

It might be more natural not to require $i \neq j$ for $\et{i}{j}$.
This is insignificant, but would have the nice effect of giving any clique a fusion-width of 1 instead of 2. Thus the simplest graphs in this measure would just be the collections of disjoint cliques.
Nevertheless, we stick to the traditional $\et{i}{j}$ operation.

\begin{definition}
 The \emph{fusion-width} $\fw(G)$ of a graph $G$ is the smallest $k$ such that there is a $k$-fusion-tree expression generating $G$.
\end{definition}

The operation $\im{i}{m}$ is introduced for convenience and to emphasize the possibility to create huge collections of identical vertices and huge bipartite graphs (together with $\et{i}{j}$). We always want to emphasize the difference between fusion-width and tree-width here.
Otherwise, except for the $\mer{i}$ operation, we have just the clique-width operations, and the slightly more efficient version of vertex creation.

Compared to clique-width, the definition of fusion-width contains the completely new operation $\mer{i}$. It is introduced to directly mimic the tree-width construction. An immediate consequence is that bounded tree-width graphs have also bounded fusion-width, without a difficult proof and without an exponential blow-up in the width parameter.
This is in sharp contrast to the relationship between bonded tree-width and bounded clique-width.

Thanks to a previous referee, we know that Courcelle and Makowsky \cite{CourcelleM2002} have already defined the parameter $\fw(G)$ when they showed that labelled graphs of bounded clique-width are closed their fusion operator. They call the parameter $\mbox{cwd}'(G)$, viewing it as an alternative notion of clique-width (justified by bounded clique-width being equivalent to bounded $\mbox{cwd}'(G)$). They don't propose $\mbox{cwd}'(G)$ to be used as a new width measure nor do they relate $\mbox{cwd}'(G)$ to tree-width.

\begin{theorem}
 Graphs with clique-width $k$ have fusion-width at most $k$. Furthermore, the number of operations does not increase from the associated $k$-expression to the associated $k$-fusion-tree expression.
\end{theorem}

\begin{proof}
 This is trivial, because all the operations of $k$-expressions are also operations of the $k$-fusion-tree expressions (with the obvious variation of replacing $i(v)$ by $\im{i}{1}$).
\qed
\end{proof}

The width parameter increases by at most 2 from tree-width to fusion-width. This tiny increase has two causes. An increase of 1 is just due to the somewhat artificial push-down by 1 in the definition of tree-width. We don't use it for fusion-width, because we want to align the measure with clique-width. The other increase by 1 is needed to have an extra label for the vertices that have already received all their incident edges.

\begin{theorem}  \label{thm:width}
 Graphs of tree-width $k$ have fusion-width at most $k+2$.
\end{theorem}

\begin{proof}
 We start with a tree decomposition with bag size $k+1$, and transform it into a $k+2$-fusion-tree expression in a bottom-up way.
One special label is reserved for vertices that have already been handled, i.e., all their incident edges have been produced. Here we refer to the $k+1$ other labels as regular labels.
 
In each bag, we use different labels for different vertices. Thus, when handling a bag, it is trivial to introduce all edges present in that bag by the corresponding $\et{i}{j}$ operations.

Choosing such a labeling is easy to do top down. We select an arbitrary node as the root of the decomposition tree, and assign distinct labels to the vertices in its bag. Before we assign labels to vertices in the bag of a node $j$ not appearing in the bag of the parent node, we assume that every vertex appearing in the bag of $j$ as well as in the bag of its parent node of the decomposition tree has already received its label. If there are still vertices in the bag of node $j$ without a label, then there are enough unused regular labels for these vertices, because we have $k+1$ regular labels and at most $k+1$ vertices in the bag of $j$.
 
 The only slightly tricky part of the fusion-tree expression is the handling of the fact that the same vertices occur in the bags of both children of a node in the tree decomposition. This is handled by using distinct vertices with the same label. When handling a node in the decomposition tree, such that more than one of its children contain the same vertex $v$ in their bags, a merge operation $\mer{i}$ is issued with $i$ being the common label of all occurrences of $v$ in the subtrees.
\qed
\end{proof}

Naturally, the following corollary is an immediate consequence.

\begin{corollary}
 If a problem can be solved in time $T(n,k)$ for graphs with $n$ vertices and fusion-width $k$, then it can be solved in time $T(n,k+2)$ for graphs with $n$ vertices and tree-width $k$.
\end{corollary}

This corollary should be compared with the corresponding important result for clique-width.

 If a problem can be solved in time $T(n,k)$ for graphs with $n$ vertices and clique-width $k$, then it can be solved in time $T(n,3 \cdot 2^{k-1})$ for graphs with $n$ vertices and tree-width $k$ \cite{CorneilR2005}.
 
 Instead of using this result, one would rather apply the corollary with its much stronger conclusion, provided that the premises are comparable. 
 
 We believe that whenever there is a nice argument that a natural problem can be solved for graphs of bounded clique-width, then we are able to nicely handle the operations $\et{i}{j}$, $\rh{i}{j}$, and $\oplus$.
Usually, we would then also have a nice argument that the problem could be solved for graphs of bounded tree-width, and we could nicely handle the operation $\mer{i}$. In such a situation, we would be able to handle all the fusion-width operations nicely, and therefore would also have an elegant algorithm whose running time should not be too bad as a function of the fusion-width.

The allowance of merging vertices with the $\mer{i}$ operation might cause two concerns. First, it is more powerful than necessary for our results. It would be sufficient to restrict it to sets of vertices of size 2. Nevertheless, we opted for the more flexible notion, because it does not cause any problems. A second concern looks more important. As the construction of graphs allows them to grow and shrink, it is reasonable to ask whether there are graphs of bounded fusion-width requiring super-polynomial size $k$-fusion-tree expressions. This is not the case

\begin{proposition}
 Every graph of fusion-width $k$ has a $k$-fusion-tree expression of size $O(|V|+|E|)$.
\end{proposition}

\begin{proof}
 Whenever some vertices are merged due to a $\mer{i}$ operation, it is possible that some edges are merged too. Assume, there is a vertex $v$ that has been used to create some edge set $E_v$ with some $\eta_{i,j}$ operations. Further assume that when $v$ is merged with some set of vertices, every edge of $E_v$ is merged with at least one other edge. Then we obtain the same graph by omitting the creation of vertex~$v$. In other words, every vertex ever created is either useless, or it is an isolated vertex in the resulting graph $G=(V,E)$, or it is responsible for at least one edge. Now, assume no useless vertices (that have no effect and disappear in a later merge operation) are ever created. Then the number of vertices ever created is at most  $|V|+|E|$. Furthermore, it is obvious, that without unnecessary label change operations $\rh{i}{j}$, the graph $G$ has a $k$-fusion-tree expression of size $O(|V|+|E|)$. \qed
\end{proof}

\section{Illustration with the Independent Set Polynomial}
Naturally, we know that finding a maximum independent set is possible in polynomial time for graphs of bounded clique-width \cite{CourcelleMR2000}. In fact the far reaching meta-theorem of Courcelle et al.\ \cite{CourcelleMR2000}
shows that this result is not just valid for the maximum independent set problem, but for every problem expressible in monadic second order logic with quantification only over sets of vertices (not edges). Furthermore, the resulting algorithm shows the problem to be in FPT with the clique-width as the parameter. 

Here, we look at a more difficult problem.
Instead of just finding the size of a maximum independent set for graphs of bounded clique-width, we count the number of independent sets of all sizes. 
We present a fixed parameter polynomial time algorithm for this counting problem.
We refer to \cite{CourcelleMR01,FischerMR08} for more discussions of the fixed parameter tractability of counting problems.

Let $[k] = \{1,\dots , k\}$ be the set of vertex labels. We define the $[k]$-labeled independent set polynomial of a $[k]$-labeled graph $G$ by
\[P(x,x_1, \dots , x_k) = \sum_{i=1}^n \sum_{(n_1,\dots,n_k) \in \{0,1\}^k} a_{i;n_1,\dots,n_k} 
	\, x^i \prod_{j=1}^k x_j^{n_j} \]
where $n_j \in \{0,1\}$ and $a_{i;n_1,\dots,n_k}$ is the number of independent sets of size $i$ in $G$ which contain some vertices with label $j$ if and only if $n_j = 1$.

We define the independent set polynomial of a graph $G$ by
\[I(x) = \sum_{i=1}^n  a_i x^i \]
where $a_i$ is the number of independent sets of size $i$ in $G$.

Then the independent set polynomial $I(x)$ can immediately be expressed by the $[k]$-labeled independent set polynomial $P(x,x_1, \dots , x_k)$.

\begin{proposition}
 The independent set polynomial $I(x)$ of a $[k]$-labeled graph $G$ is
 \[I(x) =  \sum_{(n_1,\dots,n_k) \in \{0,1\}^k} P(x,1,\dots,1)   = \sum_{i=1}^n \sum_{(n_1,\dots,n_k) \in \{0,1\}^k} a_{i,n_1,\dots,n_k} x^i. \]
\end{proposition}

\begin{theorem} \label{thm:indset}
 Given a graph $G$ with $n$ vertices and bounded fusion-width $k$, and  a polynomial size $k$-fusion-tree expression generating $G$, the independent set polynomial $I(x)$ of $G$ can be computed in FPT, i.e., in time $f(k) n^{O(1)}$ for  some function $f$.
 \end{theorem}

\begin{proof}
 We have to show how to compute the $[k]$-labeled independent set polynomial $P(x,x_1, \dots , x_k)$ of a $[k]$-labeled graph $G$.
 We compute it recursively bottom-up for the given $k$-fusion-tree expression.
 For the edgeless graph with $m$ vertices and label $i$ generated by $\im{i}{m}$, the $[k]$-labeled independent set polynomial is
 \[ x + \sum_{j=1}^m {m \choose j} x^j x_i. \]
 This polynomial is computable in time polynomial in $n$, because w.l.o.g, we can assume $m \leq n$. Otherwise, some set of vertices would be constructed together (by the $\im{i}{m}$ operation) and destroyed together (with the same merge operation $\theta_{i}$). Such a redundancy can easily be removed in a preprocessing phase.

In the following, assume for some $[k]$-labeled graph $H$, the $[k]$-labeled independent set polynomial is $\tilde{P}(x,x_1, \dots , x_k)$. 

Let $G$ be obtained from $H$ by the operation $\et{i}{j}$. Then the $[k]$-labeled independent set polynomial $P(x,x_1, \dots , x_k)$ of $G$ is obtained from the $[k]$-labeled independent set polynomial $\tilde{P}(x,x_1, \dots , x_k)$ of $H$ by deleting all monomials that are multiples of $x_i x_j$. These monomials count sets that are no longer independent after inserting all the edges between vertices labeled $i$ and $j$. 

W.l.o.g., we can assume that before a merge operation $\mer{i}$ is done, there are only 2 vertices labeled $i$. This assumption is allowed for two reasons.
\begin{enumerate}
\item 
If later a $\mer{i}$ operation is done, then every $\im{i}{m}$ operation can be replaced by an $\im{i}{1}$ operation without changing the graphs obtained after the $\mer{i}$ operation. Creating many equivalent vertices (with the same neighbors) and merging them later has the same effect as creating just one vertex.
\item
Every $\mer{i}$ operation can be replaced by a collection of $\mer{i}$ operations done immediately after a disjoint union $\oplus$ or a  relabeling operation  $\rh{i}{j}$ has created a graph with two vertices labeled $i$.
\end{enumerate}

We describe the $\mer{i}$ operation not in isolation, but only immediately after a  disjoint union $\oplus$ or a  relabeling operation  $\rh{i}{j}$. 

We now describe how to obtain the polynomial $P(x,x_1, \dots , x_k)$ of $G$ from the polynomials $P_r(x,x_1, \dots , x_k)$ of $H_r$ ($r=1,2$), where $G$ is obtained from $H_1$ and $H_2$ by the operation $\mer{i_1}\dots \mer{i_{\ell}}(H_1 \oplus H_2)$.
For ease of notation, assume that $\{x_1,\dots,x_{\ell}\} = \{x_{i_1},\dots,x_{i_{\ell}}\}$.
\begin{itemize}
\item 
Form the product $P_1(x,x_1, \dots , x_k) \cdot P_2(x,x_1, \dots , x_k)$.
\item
Delete all monomials where some $x_j$ with $1 \leq j \leq \ell$ appears with an exponent of 1. \\
If $u$ and $v$ merge into vertex $w$, then we either want both $u$ and $v$ in the independent set (to account for an independent set containing $w$), or neither (to account for an independent set not containing $w$).
\item
Replace $x_j^2$ by $x_j / x$ for $1 \leq j \leq \ell$. \\
Division by $x$ compensates the double count of a vertex in the independent set (counting $u$ and $v$ for $w$).
\item
Replace $x_j^2$ by $x_j$ for $\ell+1 \leq j \leq k$. \\
If label $j$ is not merged, then $x_j$ just indicates whether there are any vertices labeled $j$ in the independent set.
\end{itemize}

The case of a simple disjoint union ($G=H_1 \oplus H_2$) is the special case $\ell =0$ of the just described situation. Here, we just compute the product $P_1(x,x_1, \dots , x_k) \cdot P_2(x,x_1, \dots , x_k)$ and delete all monomials $x_j$ (for $1 \leq j \leq k$) appearing with an exponent 1 to obtain $P_r(x,x_1, \dots , x_k)$.

We now consider the relabeling operation $\rh{i}{j}$. First assume, there will be no succeeding $\mer{j}$ operation.
Let $G$ be obtained from $H$ by the operation $\rh{i}{j}$. Then the $[k]$-labeled independent set polynomial $P(x,x_1, \dots , x_k)$ of $G$ is obtained from $\tilde{P}(x,x_1, \dots , x_k)$ by substituting $x_j$ for $x_i$
and then replacing $x_j^2$ by $x_j$.

If $G$ is obtained from $H$ by the operation $\mer{j} \rh{i}{j}$, then we assume that in $H$ there is just one vertex labeled $i$ and one vertex labeled $j$. In this case, we proceed as follows, with the same reasoning as for the disjoint union $\oplus$ operation.
\begin{itemize}
\item 
In the given polynomial $\tilde{P}(x,x_1, \dots , x_k)$ of $H$, substitute $x_j$ for $x_i$.
\item
Delete all monomials where $x_j$ appears with an exponent of 1.
\item
Replace $x_j^2$ by $x_j$.
\end{itemize}

To prove the polynomial time claim, it is important to notice that all polynomials have at most $k+1$ variables, and all monomials are linear in each of their variables. Thus there are at most $2^{k+1}$ monomials. The number of arithmetic operations is $O(k^2 n)$, as every efficient $k$-fusion-tree expression has at most $O(k^2)$ unary operations in a row. Thus for $k$ a constant, the time is at most $O(k^2 2^{2k}n)$ if a trivial polynomial multiplication algorithm is used. With fast polynomial multiplication, based on fast Fourier transform, the time goes down to $O(k^3 2^k n)$. As the total number of independent sets is at most $2^n$, it is sufficient to do the computations with numbers of length $O(n)$. Thus, each arithmetic operation requires even with school multiplication only quadratic time. 
\qed
\end{proof}

Note that we do not claim that this was a difficult theorem. To the contrary, the point was to illustrate that a fast algorithm for a typical problem like computing the independent set polynomial restricted to bounded clique-width can be extended to a fast algorithm for this problem for graphs with the same fusion-width, i.e., for a much larger class of graphs.

\section{Relations between Tree-Width, Clique-Width and Fusion-Width}

We have $\fw(G) \leq \tw{G} + 2$ by Theorem \ref{thm:width}. 
The following inequality is trivial, as $k$-fusion-tree expressions are strictly more powerful than $k$-expressions.

\begin{proposition} \label{prop:width} \cite{CourcelleM2002}
$\fw(G) \leq \cw(G)$.
\end{proposition}

The following main result immediately implies that the graphs of bounded clique-width are exactly the graphs of bounded fusion-width. 
In fact this implication has already been shown by Courcelle and Makowsky \cite{CourcelleM2002}, as they prove the existence of a function $f$ with $\cw(G) \leq f(\fw(G))$. We still present our direct proof, because we get a much stronger result, and also because the logic framework of \cite{CourcelleM2002} might not be so widely accessible.
Many of the proof ideas are from the corresponding Theorem of Corneil and Rotics \cite{CorneilR2005}, relating tree-width to clique-width. %

\begin{theorem} \label{thm:main}
Graphs with fusion-width $\fw(G)=k$ have clique-width $\cw(G)$ at most $k 2^{k}$.
\end{theorem}

\begin{proof}
We assume, we are given a $k$-fusion-tree expression $E$ describing a graph $G$, and we want to construct a $k 2^{k}$-expression describing the same graph $G$.

We have to focus on the operation $\mer{i}$ merging all vertices labeled $i$ into one vertex. This is the only operation that has to be eliminated, because it is allowed in $k$-fusion-tree expressions determining the fusion-width, but not in $k$-expressions determining the clique-width.

We view the parse trees $T$ of the $k$-fusion-tree expression. We want to transform it into a parse tree $T'$ of a $k$-expression. The main idea is that if the vertices of some label $i$ are involved in a merge operation, we focus on the highest location $\ell$ in $T$ where such a merge occurs involving label $\hat{i}$, where $\hat{i}$ is either $i$ or a label $i$ has been changed to. 

At the corresponding location $\ell'$ in $T'$, we create the single vertex $v$ to which the vertices  labeled $\hat{i}$ have been merged by $\mer{\hat{i}}$ using the operation $\im{\hat{i}}{1}$. This means that all the vertices 
$v_1, \dots v_p$ which are finally merged into $v$ are not available further down in the tree $T'$. Therefore all operations in $T$ involving the labels of $v_1, \dots v_p$ have to be delayed until the vertex creation operation at $\ell'$.

Let $L$ be the set of labels used in $T$. The new labels in $T'$ are from $L \times \mathcal{P}(L)$, where $\mathcal{P}$ is the powerset of $L$. This way, every new label can retain its own (old) identity and in addition remember all the other old labels to which its vertices should actually be adjacent according to the edge constructing operations $\et{i}{j}$ issued in the subtree of the current node in $T$. 

Whenever a label changes due to a renaming operation $\rh{i}{j}$, in $T$, then every occurrence of $i$ (in the first or second component) of a label in $T'$ is changed to~$j$.

We say that a label $i=i_1$ is subject to a merge operation as label $\hat{i} = i_{q+1}$, if there is a sequence of label change operations $\rh{i_p}{i_{p+1}}$ ($p=1,2,\dots,q$, $q \geq 0$) (i.e., there might be no label changes and $\hat{i}=i$), such that after these changes, label $i_{q+1}$ is involved in a merge operation. 

Now, any $\et{i}{j}$ operation in $T$ is handled as follows.
\begin{itemize}
\item
If neither label $i$ nor $j$ are subject to a merge operation,
then in the corresponding location (involving several nodes) of $T'$, $\et{i'}{j'}$ is issued for all labels $i'$ with first component $i$ and labels $j'$ with first component $j$.
\item 
If label $i$ is subject to a merge operation as label $\hat{i}$, and $j$ is not subject to a merge operation, then the label $i$ is added to the second component of $j$. This way nodes labeled $j$ "remember" to create an edge to label $\hat{i}$ later. This edge is created immediately after the last merge of label $\hat{i}$.
\item
If label $i$ is subject to a last merge operation as label $\hat{i}$ after $j$ has been subject to its last merge operation, then the label $i$ is added to the second component of $j$. This way nodes labeled $j$ "remember" to create an edge to label
 $\hat{i}$ later. This edge is created immediately after the last merge of label~$\hat{i}$. \qed
\end{itemize} 
\end{proof}

\begin{corollary}
A class of graphs is of bounded fusion-width if and only if it is a class of bounded clique-width.
\end{corollary}

\begin{proof}
This follows immediately from Proposition~\ref{prop:width} and Theorem~\ref{thm:main}.
\qed
\end{proof}

\begin{corollary}
If a problem can be solved in time $O(f(k) n^c)$ for graphs of fusion-width at most $k$, then it can be solved in time $O(f(k+2) n^c)$ for graphs of tree-width at most $k$.
\end{corollary}

This immediate corollary compares favorably with the following fact.
If a problem can be solved in time $O(f(k) n^c)$ for graphs of clique-width at most $k$, then it can be solved in time $O(2^{f(k)} n^c)$ for graphs of tree-width at most $k$.

This statement cannot be much improved, as clique-width can be exponentially larger than fusion-width.

\begin{theorem} \cite{CorneilR2005} \label{thm:lower}
 For any $k$, there is a graph $G$ with $\tw(G) = k$ and $\cw(G)\geq 2^{\lfloor k/2 \rfloor - 1}$.
\end{theorem}

\begin{corollary} \label{cor:lower}
  For any $k$, there is a graph $G$ with fusion-width $\fw(G) = k$ and clique-width $\cw(G)\geq 2^{\lfloor k/2 \rfloor - 2}$.
\end{corollary}

\begin{proof}
This follows from Theorem~\ref{thm:width} and Theorem~\ref{thm:lower} 
by picking a graph with the properties of Theorem~\ref{thm:lower} and noticing that by Theorem~\ref{thm:width} its fusion-width is at most $k+2$.
Then $2^{\lfloor (k-2)/2 \rfloor - 1} = 2^{\lfloor k/2 \rfloor - 2}$ produces the result.
\qed
\end{proof}

Corollary \ref{cor:lower} shows that our example of the independent set polynomial proves the fusion-width to be a powerful notion. There are graphs with fusion-width $\fw(G) = k$ and clique-width $\cw(G)\geq 2^{\lfloor k/4 \rfloor - 2}$. If for such a graph, we have a $k$-fusion-tree expression,  then we can compute its independent set polynomial in time $O(k^3 2^k n)$ by the method of Theorem~\ref{thm:indset}. Using just the fact that the clique-width $\cw(G)\geq 2^{\lfloor k/4 \rfloor - 2}$, we would be unlikely to find a better algorithm based on clique-width. Thus we would only have an 
an algorithm that is doubly exponential in $k$ for computing the independent set polynomial of these graphs.

We believe that the independent set polynomial is not an isolated instance. It has merely been used to illustrate the convenience and power of the fusiion-with parameter. Many other examples could be used instead.

\section{Conclusion}
We have introduced a new width measure, the fusion-width. Its purpose is two-fold. It provides a tool for handling generally difficult problems for a large class of graphs. It also sheds a light on the essence of the generalization from bounded tree-width to bounded clique width. It is the ability at each stage of the construction not only to add edges between a limited number of vertices, but to add complete bipartite graphs between a limited number of sets of vertices.

\section{Open Problems}  %
What is the complexity of determining the fusion-width of a graph? Is it in XP (fixed parameter polynomial time) or even in FPT (fixed parameter tractable)? How well can fusion-width be approximated?

Find interesting classes of graphs with a large clique-width and a small fusion-width.

What is the relationship between fusion-width, rank-width, and boolean width?

\section{Acknowlegement}
The help of previous anonymous reviewers has improved this paper significantly.

\end{document}